\documentclass{article}
\usepackage[T2A]{fontenc}
\usepackage[utf8]{inputenc}
\usepackage[russian,english]{babel}
\usepackage{amssymb} 
\usepackage{amsmath} 
\usepackage{amsthm}
\usepackage{url}
\usepackage{graphicx}

\author{N. N. Vasiliev\thanks{St. Petersburg Department of V.\,A.\,Steklov Institute of Mathematics of the Russian Aca\-de\-my of Sciences, Saint Petersburg Electrotechnical University, \texttt{vasiliev@pdmi.ras.ru}.}
\ and D. A. Pavlov\thanks{Institute of Applied Astronomy of the Russian Academy of Sciences, \texttt{dpavlov@iaaras.ru}}}
\title{The computational complexity of the initial value problem
 for the three body problem}
\date{}

\newtheorem{Def}{Definition}
\newtheorem{Thm}{Theorem}
\newtheorem{Lem}{Lemma}

\begin{document}
\maketitle

{\small{\noindent The final publication is available at Springer via\\
\url{http://doi.org/10.1007/s10958-017-3407-3}}}
\begin{abstract}
\noindent
The paper is concerned with the computational complexity of the
initial value problem (IVP) for a system of ordinary dynamical
equations. Formal problem statement is given, containing a Turing
machine with an oracle for getting the initial values as real
numbers. It is proven that the computational complexity of the IVP for
the three body problem is not bounded by a polynomial. The proof is
based on the analysis of oscillatory solutions of the Sitnikov
problem that have complex dynamical behavior. These solutions
contradict the existence of an algorithm that solves the IVP in
polynomial time.
\end{abstract}

\section{Introduction}
The problem of numerical integration of ODE systems is undoubtedly
one of the most popular problems in applied mathematics. There exists a
huge number of algorithms and program packages for obtaining numerical
solutions of systems of differential equations originating from
math, physics, celestial mechanics and engineering. However,
there is little available research in the area of computational
complexity of the initial value problem itself
(some results are obtained in \cite{kawamura} and~\cite{reif}).
In most other works, complexity of particular algorithms is
analyzed, in terms of either the number of basic arithmetical
operations performed on each step, or the number of calls to
first or higher-order derivatives.

In this work, a formal statement is presented of the IVP for a system of ODEs.
In that statement, the input data for a problem will be:
the initial conditions, the point $t$ in time,
and a precision $\varepsilon$.
An algorithm is supposed to consume the input and produce
the output (an approximate state of the system at time $t$)
that matches the actual state of the system up to $\varepsilon$.
It must be noted that since the said statement includes real
numbers, we can not work with just data of finite length.
While it is sufficient to treat $t$ and $\varepsilon$ as rationals,
initial conditions are a different story: there is no prior
knowledge of how many digits in them will be sufficient to
ensure that the solution at $t>0$ will be obtained with precision
$\varepsilon$.

There are several approaches to work around that difficulty. The first
is to consider an infinite input tape (or several infinite tapes)
whose cells contain the digits of the initial conditions.
A Turing machine for the given IVP can read the digits on demand.
The second approach is to have a Turing machine use an oracle
that gives the needed digits on demand. The third approach
is to use a secondary Turing machine that prints out the digits
into the tape by request from the main Turing machine.

The third approach, as opposed to the first two, is that it would
limit us to just the constructive real numbers. In this work,
the second approach (with an oracle) is used. It differs from the
first one in the conventions of complexity analysis: the calls
to an oracle account for the time complexity of algorithm as a function
from the (finite) input length, while in the first approach
makes the input infinite, rendering the complexity analysis difficult.

Another obstacle in the formal statement of the problem is the
following: even if the derivatives in the system of ODEs are
known Lipschitz-continuous functions, the problem of existence
of the ODE solution at point $t$ can be undecidable.

We will show that even in the provably decidable case, the
complexity of the IVP can be non-polynomial. The proof is
based on the investigation of systems with complex dynamical
behavior. As a basic example, we will use the classical Sitnikov problem
for a three-body gravitational system, where two bodies follow
elliptic orbits on a plane, and the third body stays on the line perpendicular
to that plane. In the general case, the third body
does unending oscillations with arbitrary amplitudes.

Instead of the oscillating solution of Sitnikov problem, we could use other
dynamical systems exhibiting complex behavior, like, for instance, a
neighborhood of some homoclinic solution. Their computational complexity
would have turned non-polynomial, too. 

In this work, we do not use a natural representation of such
solution in the terms of symbolic dynamics~\cite{alexeyev}.
Rather, to prove the absence of a polynomial algorithm for our
formal IVP statement, it is sufficient to show that for a certain
neighborhood of initial conditions in phase space, the number of
algorithmically distinguishable trajectories is exponential in $t$.

\section{Turing machine for the initial value problem}\label{turingsec}

\def\LENGTH{\textrm{LENGTH}}
\def\xvec{\mathbf{x}}
\def\fvec{\mathbf{f}}
\def\R{\mathbb{R}}
\def\Q{\mathbb{Q}}
\def\N{\mathbb{N}}
\def\Z{\mathbb{Z}}
\def\d{\mathrm{d}}
\def\qvec{\mathbf{q}}
\def\pvec{\mathbf{p}}
\def\vvec{\mathbf{v}}
\def\P{\mathcal{P}}

We estimate the computational complexity of the initial value problem
for the dynamical system

\begin{equation}\label{ode}
\begin{array}{lcl}
  \dot \xvec & = & \fvec(\xvec) \\
  \xvec(0) & = & \xvec_0
\end{array}
\end{equation}
where $\xvec \in D$, $\xvec_0 \in D$ is a real vector, and
$\fvec: D \to \R^n$ is a computable real vector-valued function
(open set $D\subseteq \R^n$ is the phase space of the system).

This work deals with the case when the solution
$\xvec^*(t) : \R \rightarrow D$:
\begin{enumerate}
\item exists on the whole $\R$;
\item is unique;
\item is a computable real vector-valued function.
\end{enumerate}

Solutions that do not extend to $\R$ are called \textsl{singular}.
The problem of determining the singularity of a solution is undecidable
(see section \ref{known-nbody}).
Uniqueness of a solution, it it exists, is guaranteed given that
the function $\fvec$ is locally Lipschitz-continuous in every point in $D$.
(The proof of that fact can be found e.g. in~\cite[p. 15]{burke}.)
However, the local Lipschitz-continuity does not imply the existence
of the solution on $\R$.

If $\fvec$ is defined on $D$ when $D=\R^n$ and is (globally)
Lipschitz-continuous, then the solution on $\R$ does exist
for all $\xvec_0$ and is unique due to the Cauchy-Lipschitz theorem.

If $\fvec$ is continuous at every point in $D$, then every unique
solution is computable by a (non-practical) combinatorial
algorithm~\cite{collins}. In particular, that holds for any computable
$\fvec$, since every computable function is continuous.

In~\cite{repin}, it is proven that the solution of an IVP is computable
with a modification of Picard--Lindel\"of method, if $\fvec$
is Lipschitz-continuous on $D$. This important fact is quite non-trivial,
despite the existence of hundreds numerical integrators for ODE.
The vast majority of these integrators suffer from \textsl{saturation}:
the step size being small enough, the error grows upon further
decrease of the step size. Therefore, these
integrators can not in principle obtain a solution up to an arbitrary
precision~\cite{babenko}.

To summarize: with $D=\R^n$ and Lipschitz-continuous $\fvec$, the solution
of (\ref{ode}) with any $\xvec_0$ exists on $\R$, is unique and computable.
It follows independently from~\cite{collins} and~\cite{repin}. In both
sources, the computability is proven for the solution being
the function of $\xvec_0$ and $t$, rather than just $t$.

In this work, we limit ourselves with the study of a particular
instance of the three-body problem (see Section~\ref{sitnikov-problem}).
The subject for study is the asymptotic dependence of the computational
complexity of the solution $\xvec^*(t)$ on the value of $t$; the dependence
on the precision of $t$ is not considered. In the text that follows,
$t$ in the IVP is treated as rational, while $\xvec_0$ is a real vector.
The complexity analysis of another special case of IVP, where
$t \in \R$, is given in~\cite{kawamura}.

\begin{Def}\label{ivpfunc}
The solution function of an initial value problem (\ref{ode}) is
the function $S(\xvec_0, t): D \times \Q \to D$, where $S|_{\xvec=\xvec_0}: \Q \to D$ is a computable real vector-valued function, whose closure on the real
axis is the solution of (\ref{ode}).
\end{Def}

\begin{Def}\label{turing}
  Turing machine that computes the solution function of an IVP is
  a Turing machine that accepts rational $t$ and $\varepsilon$ as input;
  has an oracle $\varphi$ that instruments $\xvec_0$ as a computable
  real vector; and produces the value of the solution $\xvec(t)$ corresponding
  to given $\xvec_0$ and $t$, with the precision $\varepsilon$.
\end{Def}

It should be noted that in terms of complexity theory, the IVP belongs to
the class of \textsl{function problems}, as opposed to more studied
\textsl{decision problems}. The job of the oracle in the Turing machine 
is to write into its tape the representation of $\xvec_0$ up to
an arbitrary precision, specified by the machine itself. It is obvious
that the time required by the Turing machine includes the time
to read the oracle tape.

\begin{Def}
The IVP (\ref{ode}) has polynomial complexity if there exists a 
Turing machine from the definition \ref{turing} that computes
its solution function in time bounded by $\P(\LENGTH(t), \LENGTH(\varepsilon))$,
where $\P$ is an arbitrary polynomial.
\end{Def}

\noindent \textbf{Remark.} Without loss of generality, it can be assumed
that $\varepsilon = 2^{-l}$, hence $\LENGTH(\varepsilon) = l$.

\begin{Def}
  Suppose A and B are two IVPs. A is called polynomially reducible
  to B if there exist the following functions, computable in polynomial time:
  $G : D^{(A)} \to D^{(B)}$ and $H : D^{(B)} \to D^{(A)}$,
  so that for any initial state $\xvec_0^{(A)}\in D^{(A)}$
  and a corresponding solution $\xvec^{*(A)}(t)$ the following holds:
  $\xvec^{*(A)}(t) = H(\xvec^{*(B)}(t))$, where
  $\xvec^{*(B)}(t)$ is a solution of B with initial state
  $\xvec_0^{(B)} = G(\xvec_0^{(A)})$.
\end{Def}

\noindent \textbf{Statement.} If IVP A is polynomially reducible to IVP B,
and B has polynomial complexity, then A has polynomial complexity as well.

\section{Analysis of the computational complexity of the IVP for the three-body problem}
\subsection{$N$-body problem}
Gravitational $N$-body problem is concerned with the Newtonian motion of $N$
point-masses in three dimensions. The system of ODEs for this problem is
the following::
\begin{equation}\label{nbody}
\left. \begin{array}{rcl}
 \dot \pvec_i & = & \vvec_i, \quad i = 1..N\\
 \dot \vvec_i & = & \sum\limits_{\substack{j=1\\j\neq i}}^N\mu_j \frac{\pvec_j-\pvec_i}{|\pvec_j-\pvec_i|^3}, \quad i = 1..N
\end{array} \quad \right\} 
\end{equation}
where $\mu_i\in\R$, $\mu_i \geq 0$, $\pvec_i \in \R^3$, $\vvec_i \in \R^3$.

With $N=3$, the initial state of the system is given by a 21-vector
$\xvec_0=(\mu_1,\mu_2,\mu_3,p_{1,1},\ldots,p_{3,3},v_{1,1},\ldots,v_{3,3})$,
while the system (\ref{nbody}) defines a computable real vector-valued
function $\dot\xvec = {\mathbf f}(\xvec)$.
(The first three variables do not depend on $\xvec$ or $t$.)

\subsection{Known results}\label{known-nbody}
The classical two-body problem ($N=2$) has a solution in algebraic
functions of initial state and $t$. Depending on the configuration
of the system, the two bodies follow either a Keplerian orbit
(a parabola, hyperbola, or ellipse) or move along a line.
The detailed description of the solutions can be found in multiple sources.
Given those algebraic solutions, it is not difficult to show
that the IVP for a nonsingular two-body problem has polynomial complexity.

With $N=3$ the problem does not have a generic algebraic solution, as proven
by Poincar\'e. However, Sundman in 1912 derived a solution in the form of
converging series. Unfortunately, the estimate of the number of terms
required to calculate the series at point $t$ with a sensible precision
is exponential in $t$~\cite{beloriszky}. Merman improved
Sundman's result and found other series~\cite{merman}, though still
exponential in $t$.

In practical tasks related to the $N$-body problem (in particular, in ephemeris
astronomy) algorithms of numerical integration are used to obtain
approximate solutions. The time complexity of such algorithms
has a fundamental lower bound of $O(t)$, hence it can not be upper-bounded
by a polynomial of $\LENGTH(t)$.

The bottom line is that the known algorithms for the IVP for the three-body
problem are non-polynomial.  However, that does not disprove the polynomial
complexity of the problem.

\medskip

On a different note, let us show that there is a singular solution of the
$N$-body problem that has a nonsingular one in any neighborhood. Let $N=2$. Two
bodies collide if they are thrown upon each other along a straight line, while a
smallest deviation from the straight line will prevent the collision (if the
velocity is big enough). This implies the undecidability of the problem of
determination of singularity with computable real $\xvec_0$: it requires the
solution of equality relation of real numbers which does not exist.

\subsection{Sitnikov problem}\label{sitnikov-problem}
From now on, we will focus on a special case of the three-body problem,
where two of the bodies are of equal positive mass, while the third body
is massless and lies on a line, perpendicular to the plane of the motion
of the first two bodies and passing through their center of mass 
(Fig.~\ref{fig-sitnikov}). Hence, the two bodies follow the unperturbed
(Keplerian) orbit; in this problem, the elliptic orbit is the case.

Let us place the center of mass at the origin, and the $Z$ axis along
the line where the third body is. Let us denote $r(t)$ the distance
from the first body (and the second, as their trajectories are symmetric)
to the origin.

\begin{figure}[h]
  \centering
  \caption{Sitnikov problem}
  \vskip 3mm
  \includegraphics[width=\textwidth]{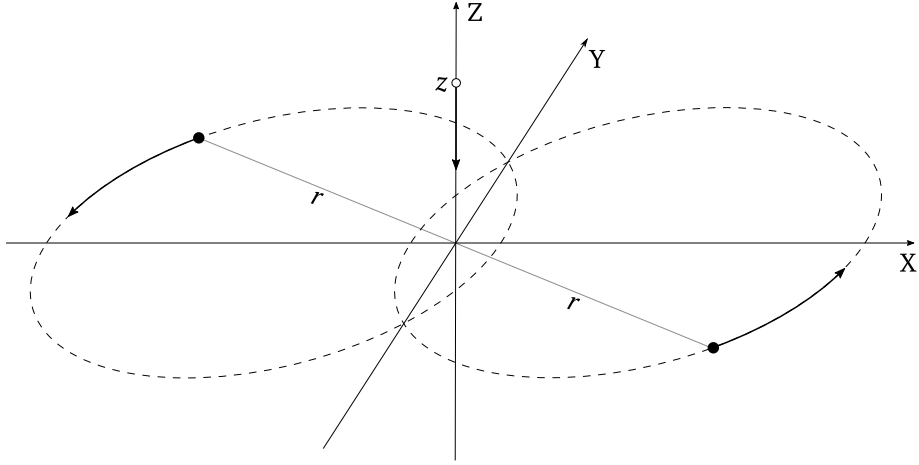}
  \label{fig-sitnikov}
\end{figure}

Following Newtonian laws (\ref{nbody}), the coordinate of the third body,
denoted as $z$, obeys the following differential equation:
\begin{equation}\label{sitnikov}
\ddot z = -\frac{2 \mu z}{\sqrt{z^2 + r(t)^2}^3},
\end{equation}

\noindent where $\mu$ is the gravitational constant of the first and second bodies.
Periodic function $r(t)$ comes from the solution of the two-body problem:

\begin{equation}\label{twobody}
 \begin{array}{rcl}
   r(t) &=& a(1-e \cos E(t)) \\ 
   E(t) - e \sin E(t) &=& \sqrt\frac{\mu}{a^3}(t - t_0)
   \end{array}
\end{equation}

$a$ (semimajor axis), $e$ (eccentricity) and $t_0$ (epoch) are constants that
can be calculated from the initial state of the two bodies. $E(t)$ is the
eccentric anomaly angle. The period of $r(t)$ is $P = 2\pi \sqrt\frac{a^3}{\mu}$.

The initial values in the Sitnikov problem are:

\begin{itemize}
\item $a > 0$, $e \in (0..1)$, $\mu > 0$ --- parameters of the orbit of the two bodies;
\item $z_0 = z(0)$ --- initial position of the third body in the $Z$ axis.
\item $v_0 = \dot z(0)$ -- initial velocity of the third body in the $Z$ axis.
\item $\phi = E(0)$, $0 \leq \phi < 2 \pi$  --- initial value of the eccentric anomaly of
  the orbit of the two bodies.
  \end{itemize}
 
The state vector of the system is accordingly
$\xvec = (a,e,\mu,z,v,E)$. $a$, $e$ and $\mu$ do not depend on time;
$\dot z = v$; $\dot v = \ddot z$ from (\ref{sitnikov});
$\dot E$ follows from (\ref{twobody}):
\begin{equation}\label{sitnikov-xdot}
\begin{array}{rcl}
\dot \xvec & = & \fvec(\xvec) = (0, 0, 0, v, \ddot z, \dot E) \\
\ddot z & = & -\frac{2\mu z}{\sqrt{z^2+a^2(1 - e\cos E)^2}^3} \\
\dot E & = & \frac{\sqrt{\mu a}}{1 - e \cos E}
\end{array}
\end{equation}

\noindent \textbf{Statement.} IVP for the Sitnikov problem (\ref{sitnikov}) is
polynomially reducible to the IVP for the three-body problem (\ref{nbody}).

The study of the trajectories of $z(t)$ in this system was started by
Kolmogorov, while Sitnikov was the first to prove the existence of
the oscillatory motions in this system~\cite{sitnikov1960}. His proof was also
the first proof of this kind for three-body systems in general.

\begin{Thm}\label{sitnikov-lip} In the Sitnikov problem, there are no singularities,
  and the function $\fvec$ is Lipschitz-continuous on the whole domain.
\end{Thm}
\begin{proof}
From Eqs. (\ref{sitnikov}) and (\ref{twobody}), along with the fact that
$r(t) > 0$, instantly follows that $\fvec$ is defined and continuous with any
$z, v, E \in \R$.

Let us prove the Lipschitz-continuity of $\fvec$ by showing that all its
partial derivatives w.r.t. $\xvec$ are bounded. We write down those derivatives,
skipping the zero ones:
\begin{eqnarray}
\partial v / \partial v &=& 1 \label{partial1}\\
\partial \ddot z / \partial z &=& -2\mu\left(\frac{1}{w^3}-\frac{3z^2}{w^5}\right)\label{partial2}\\
\partial \ddot z / \partial E &=& -3\mu z\frac{2a^2(1-e\cos E)\sin E}{w^3}\label{partial3}\\
\partial \dot E / \partial E &=& -\sqrt{\mu a}\frac{e \sin E}{(1 - e \cos E)^2}\label{partial4}
\end{eqnarray}
\noindent (Notion $w=\sqrt{z^2+a^2(1 - e\cos E)^2}$ is used for brevity.)

It is evident that all those functions are defined and continuous for any $z,v,E\in\R$
(for (\ref{partial4}) it is important that $0<e<1$).
The boundedness of (\ref{partial1}) and (\ref{partial4}) is trivial.
The boundedness of (\ref{partial2}) follows from the fact that it approaches
zero as ${z\to\pm\infty}$: $\frac{1}{w^3}\to 0$ and $\frac{z^2}{w^5}\to 0$.
Similarly, (\ref{partial3}) is bounded because $\frac{z}{w^3}\to 0$ at ${z\to\pm\infty}$.
\end{proof}

Existence, uniqueness, and computability of the solution of the IVP for the Sitnikov
problem follow from Theorem \ref{sitnikov-lip} and the references given in Section \ref{turingsec}.

For the rest of the article, we consider the Sitnikov problem
with $z_0 = 0$, omitting the solutions where the third body never crosses the plane.

\subsection{Combinatorial properties of the solutions of the Sitnikov problem}
Sitnikov's result about the oscillatory motion was significantly extended
by Alexeyev, who not only discovered the existence of all the classes of final
motions in this problem, but also proved the following~\cite{alexeyev1,alexeyev2,alexeyev3}:
\begin{Thm}
  For any sufficiently small eccentricity $e > 0$ there exists an $m(e)$ such that
  for any double-infinite sequence $\{s_n\}_{n\in\Z}, s_n \geq m$
  there exists a solution $z(t)$ of the equation (\ref{sitnikov})
  whose roots satisfy the equation
\begin{equation}\label{sktau}
\left\lfloor \frac{\tau_{k+1}-\tau_k}{P}\right\rfloor = s_k,\ \forall k \in \Z.
\end{equation}
\end{Thm}

The shortened version of the original theorem is given, excluding the
finite and semi-infinite sequences. Alexeyev also proved a generalization
of his theorem to the case when the third body has a nonzero mass.
A simpler proof was later obtained by Moser~\cite{moser}.

In what follows, we restrict our analysis to $t\geq 0, k\geq 0$ ($\tau_0 = 0$).

\begin{Lem}\label{lem1} Let $C(T)$ be the set of (finite) sequences of the
  form $(s_1,\ldots, s_k)$, $s_i \geq m > 1,\ s_i\mod 2 = 0,\ m \mod 2 = 0$,
  for each of which \textsl{any} sequence
  $(\tau_0,\ldots,\tau_{k+1})$ satisfying (\ref{sktau}) lies in the interval $[0, T]$
  (i.e. $\tau_{k+1} \leq T$). $|C(T)|$ has an asymptotic lower bound exponential in $T$.
\end{Lem}

\begin{proof}
  Obviously, $C((m+1)P)=1$. For some $T\geq(m+1)P$, let us consider the
  interval $[T, T + (m+1)P]$. Any sequence $(s_1,\ldots s_k)\in C(T)$ can be
  extended to a sequence from $C(T+(m+1)P)$ by the following ways:
  \begin{itemize}
    \item $(s_1,\ldots, s_k,m)\in C(T+(m+1)P)$
    \item $(s_1,\ldots, s_k+2i)\in C(T+(m+1)P),\ \forall 0<i\leq m/2$
  \end{itemize}
  Consequently, $|C(T+(m+1)P)| \geq (m/2+1)|C(T)|$,
  and that implies $|C(T)|\geq (m/2+1)^{\frac{T}{(m+1)P}}$ for sufficiently large $T$.
  If $m>0$, this bound is exponential in $T$.
  \end{proof}

\subsection{Computational complexity of the IVP for the Sitnikov problem}
We give two lemmas that describe important properties of $z(t)$.
The first lemma gives a lower bound of $|z(t)|$ between two roots separated
by a certain distance. In the proof of the lemma, the Sturm's comparison theorem
is used:

\begin{Thm}[Sturm's comparison theorem]
Consider two equations:
\begin{equation}\label{sturm-q}
\ddot x = - q(t) x
\end{equation}
and
\begin{equation}\label{sturm-Q}
\ddot x = - Q(t) x,
\end{equation}
where $q$ and $Q$ are continuous functions.
Let a nonzero solution of (\ref{sturm-q}) $x(t)$ has roots $a$ and $b$,
and $Q(t) > q(t)$ on $t\in[a,b]$. Then any solution of (\ref{sturm-Q})
has a root on $(a,b)$.
\end{Thm}

\begin{Lem}\label{lem2}
  Let $z^*(t)$ be a solution of the Sitnikov problem (\ref{sitnikov}) with
  initial values $a, e, \mu, \phi, v_0$. According to the previous assumptions,
  let $z^*(0) = 0$. To be specific, we consider $v_0>0$ (the case of negative
  $v_0$ is a mirroring of that). Let $\tau$ be the smallest positive root of  $z^*$.
  Then $\exists t \in (0, \tau) : z^*(t) \geq h$, where
\begin{equation}\label{hdef}
h = H(\tau) = \sqrt{\left(\frac{2\mu\tau^2}{\pi^2}\right)^{\frac{2}{3}}-a^2}
\end{equation}
\end{Lem}
\begin{proof}[Proof by contradiction]
  Suppose $z^{*}(t)<h$, $0\leq t \leq\tau$.
  Since $z^*$ is the solution of (\ref{sitnikov}), then it is also the solution
  of the following equation:
  \begin{equation}\label{ddotzstar}
  \ddot z = -\frac{2 \mu z}{\sqrt{{z^*(t)}^2 + r(t)^2}^3},
  \end{equation}
  where the factor of $z$ depends only on $t$, but not on $z$. Let us denote this factor $Q(t)$:
  \begin{equation}\label{zQ}
    \ddot z = - Q(t) z.
  \end{equation}
  Since $z^{*}(t)<h$ by the assumption, and $r(t)\leq a$, then
  $$Q(t) > \frac{2 \mu}{\sqrt{h^2 + a^2}^3}$$
  Denoting
\begin{equation}\label{qdef}
  q=2 \mu / \sqrt{h^2 + a^2}^3,
\end{equation}
  we write a differential equation
  \begin{equation}\label{zq}
    \ddot z = - q z.
  \end{equation}
  Since $q>0$ the equation (\ref{zq}) is the equation of a harmonic oscillator.
  We examine its solution $z^{**}$ for initial conditions $z(0) = 0, \dot z(0) = v_0$:
  $$z^{**}(t) = v_0 \sin(\sqrt q t)$$
  By the Sturm's comparison theorem, between two roots of
  $z^{**}$---0 and $\pi/\sqrt{q}$---there exist roots of any solution of 
  (\ref{zQ}), including $z^{*}$.
  Since $\tau$ was chosen as the smallest positive root of $z^{*}$, it must be that
  $\tau<\pi/\sqrt{q}$.
  However, by construction of $q$ (\ref{qdef}) and $h$ (\ref{hdef})
  it follows $\tau=\pi/\sqrt{q}$, hence the contradiction.
  \end{proof}

\begin{Lem}\label{lem3}
Consider a nonnegative function $z(t)$, continuous and convex on $[t_1, t_2]$; let $z(t_1) = z(t_2) = 0$;
let at some $t\in[t_1,t_2]$ $z(t)>h>0$. Then
$\exists t_a, t_b \in [t_1, t_2] : (t_b - t_a) > \frac{3}{4}(t_2 - t_1), \forall t\in (t_a, t_b)\ z(t) > h/4$.
\end{Lem}
\begin{proof}
  $z(t)$ has one (strict) maximum at $(t_1,t_2)$, let us say that $t_3$ is the point where
  the maximum is reached. Let us place points (Fig.~\ref{fig-zconvex}): A$(t_1, 0)$, B$(t_3, z(t_3))$, C$(t_2, 0)$. Let the line $z=h/4$ cross $AB$ at point $D$ and $BC$ at point $E$.
  Similarly, let the same line cross the $z(t)$ curve at $F$ and $G$.

\begin{figure}[h]
  \centering
  \caption{Example for Lemma \ref{lem3}}
  \includegraphics[width=0.8\textwidth]{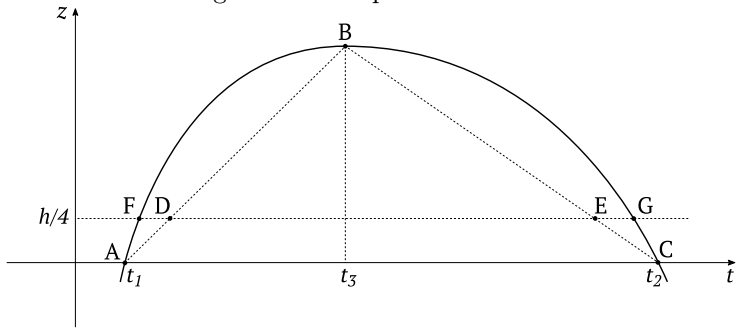}
  \label{fig-zconvex}
\end{figure}
  
Since $z(t)$ is convex, it lies above ABC, with the exception of A, B and C
themselves (Fig. \ref{fig-zconvex}). Consequently,
$|\textrm{DE}| < |\textrm{FG}|$. At the same time, from the similarity
of triangles it follows that 
  $\frac{|\textrm{DE}|}{|\textrm{AC}|} = 1-\frac{h/4}{z(t_3)}$. Since $z(t_3) > h$
and $|\textrm{AC}| = (t_2 - t_1)$, we get $|\textrm{FG}| > \frac{3}{4}(t_2 - t_1)$.
The horizontal coordinates of $F$ and $G$ are the desired $t_a$ and $t_b$.
\end{proof}

\begin{Thm}\label{main}
  The time complexity of an initial value problem for the Sitnikov problem
  with any fixed value of eccentricity does not have a polynomial upper bound.
\end{Thm}
\begin{proof}[Proof by contradiction]
  Suppose that there exists a Turing machine $M$ that calculates the solution
  function of the IVP for the Sitnikov problem in time
  $\P(\LENGTH(t), \LENGTH(\varepsilon))$, where $\P$ is arbitrary polynomial.
  
  We examine the solutions at the interval $t\in[0, T],\ T\in\N$.
  From Lemma \ref{lem1} and Alexeyev's theorem, the number $C(T)$ of different
  solutions $z(t)$, forming different sequences
   $(s_1,\ldots,s_k)$ with $s_k \mod 2 = 0, s_k \geq m\ (m \mod 2 = 0)$,
  has a lower bound of $(m/2+1)^\frac{T}{(m+1)P}$,
  where $m$ depends only on $e$. (The Alexeyev's theorem allows zero and odd $m$,
  but we can round the $m$ up to be a nonzero even number, without trouble to
  the theorem.

  We build an algorithm for \textsl{recovery} of the sequence $(s_1,\ldots,s_k)$
  that corresponds to a solution $z(t)$ for some initial values, using our supposedly
  existing Turing machine $M$. We choose the parameters $\delta\in\Q, \delta < mP/2$
  and $\varepsilon = 2^{-l} (l\in\N), \varepsilon < h/4$, where $h = H(mP)$. (Note
  that $P$ is a computable real number.) 
  
  Let us build on $[0, T]$ a uniform grid with a step $\delta$;
  on each node $\{t_i = i\delta, 0 < i \leq \lfloor T/\delta\rfloor\}$
  we can compute the state of the system up to the precision $\varepsilon$.
  The grid has the following important properties:
  \begin{itemize}
  \item If $|z(t_i)| > h/4$, then from Lemma \ref{lem3} follows that the closest
    root to $t_i$ lies no farther than  $mP/4$.
    \item From above it follows that two neighbor nodes can not both have
      $|z|<h/4$
    \item Calculated $z(t_i)$ can be divided into three classes:
      positive ($z > 0$ for sure), negative ($z < 0$ for sure)
      and undefined (the sign of $z$ is not determined within the given precision).
    \item Positive and negative nodes can go any number in a row,
      while there can be only one undefined node in a row.
    \item From the estimate of the distance between roots, it is evident that
      if there are no nodes between a positive node and a negative node,
      or if there is (one) undefined node, then $z$ has exactly one root
      in between.
  \end{itemize}
  Given that the $s_k$ are even, it is easily seen that  $p$ nodes in a row of the same sign
  correspond to $s_k = \lceil (p+1)/2 \rceil$; undefined nodes do not correspond to any
  $s_k$.

  It is not important how long it took to recover the sequence of $s_k$.
  What matters is that all the ``calls'' to out Turing machine $M$
  have used \textsl{the same oracle} for the computation of the (same) initial state.
  But, as we supposed, $M$ did not have a chance to read more than
  $\P(\LENGTH(t_i), \LENGTH(\varepsilon))$ digits from the oracle tape for any $t_i$,
  which is no more than $P(\log_2 T, l)$; hence, basing on what it had read,
  it can possibly generate no more than $2^{\P(\log_2 T, l)}$ different outcomes.
  At the same time, we proved that our algorithm recovers any of at least 
  $(m/2+1)^{T/((m+1)P)}$ sequences, which (as $m>0$) is not bounded by the said polynomial.
  \end{proof}

\section{Conclusion and future work}
In this work we examined the theoretical complexity of the initial value
problem. We have shown that the lower time bound of that complexity
can not be polynomial for the three-body problem (instantly meaning the
absence of such a bound for the $N$-body problem). The choice of
the three-body problem and oscillatory trajectories is not principal.
We believe that similar results can be obtained in other systems,
where, with the help of methods of symbolic dynamics, complex
dynamical behavior can be shown and analyzed. We already mentioned
homoclinic trajectories, discovered by Poincar\'e for the three
body problem. It seems appropriate to quote his work ``New methods
of celestial mechanics''~\cite{poincare} here:

``One is struck by the complexity of this figure I am not even
attempting to draw. Nothing can give us a better idea of the
complexity of the three-body problem and of all problems of dynamics
where there is no holomorphic integral and Bohlin’s series diverge.''

On a different note, for the integrable dynamical systems---those who
have computable integrals of motion with good complexity bounds in $t$ and $\varepsilon$---
it is possible to derive complexity bounds for the initial value problem
in our formal statement. Those bounds will be polynomial by
$\log(t)$ and $\log(1/\varepsilon)$. That can point to a link between
computational complexity of the IVP and integrability.

On another different note, in this work the computational complexity
of the IVP is examined at the ``macro level'' (rational $t\to\infty$),
but what is left aside is the ``micro level'' (real $t$), where
the precision of $t$ plays an important role~\cite{kawamura}.
Another work is planned devoted to that case.

\end{document}